\documentclass[pra,preprintnumbers,amsmath,amssymb,superscriptaddress,twocolumn,nofootinbib, showpacs]{revtex4}
\usepackage{bm}
\usepackage{graphicx}
\usepackage{amsbsy}
\usepackage{amsmath}
\usepackage{amsfonts}
\usepackage{bbm}
\usepackage{amsthm}
\usepackage{color}
\usepackage{mathrsfs}
\usepackage{yfonts}
\usepackage{eufrak}

\begin{document}

\theoremstyle{plain}
\newtheorem{theorem}{Theorem}
\newtheorem{lemma}[theorem]{Lemma}
\newtheorem{conjecture}[theorem]{Conjecture}
\newtheorem{proposition}[theorem]{Proposition}
\newtheorem{corollary}[theorem]{Corollary}

\theoremstyle{plain}
\newtheorem{definition}{Definition}
\newtheorem{exmp}{Example}
\newtheorem{example}{Example}[section]
\theoremstyle{remark}
\newtheorem*{remark}{Remark}

\title{Constructing Monotones For Quantum Phase References  In Totally Dephasing Channels}
\author{Borzu Toloui}
\email{btoloui@qis.ucalgary.ca}
\affiliation {Institute for Quantum Information Science, University of Calgary, Alberta, T2N 1N4, Canada.}
\author{Gilad Gour}
\affiliation {Institute for Quantum Information Science, University of Calgary, Alberta, T2N 1N4, Canada.}
\affiliation{Department of Mathematics and Statistics, University of Calgary, Alberta T2N 1N4, Canada.}
\author{Barry C. Sanders}
\affiliation{Institute for Quantum Information Science, University of Calgary, Alberta, T2N 1N4, Canada.}
\pacs{03.67.-a, 03.65.Ud}
\begin{abstract}
Restrictions on quantum operations give rise to resource theories.  Total lack of a shared reference frame for transformations associated with a group~$G$ between two parties  is equivalent to having, in effect,  an invariant channel between the parties and a corresponding superselection rule (SSR). The resource associated with the absence of the reference frame is known as `frameness' or `asymmetry'.    
We show that any entanglement monotone for pure bipartite states can be adapted as
a pure-state frameness monotone for phase-invariant channels
(equivalently U(1) superselection rules)
and extended to the case of mixed states via the convex-roof extension.
As an application,
we construct a family of concurrence monotones for U(1)-frameness  for 
general finite-dimensional Hilbert spaces.
Furthermore we study `frameness of formation' for mixed states analogous to entanglement of formation.  
In the case of a qubit, we show that it can be expressed as an analytical function of the concurrence analogously to the Wootters's formula for entanglement of formation.
Our results highlight deep links between entanglement and frameness resource theories.
\end{abstract}

\maketitle
\section{Introduction}
Effective communication is achieved by sending objects through a channel,
with the information encoded into the physical states of the objects.
For example,  information can be encoded into energy levels, rotations, or arrival times of the objects,
but every method of physically transmitting information requires a shared reference frame (RF) between
the parties to the communication.  
For the cases of energy levels, rotations, and arrival times, 
the parties should share the ability to read energies, orientations, and the passage of time in order to have the
ability to decode the encoded information,
and this shared ability is achieved by calibrating their instruments.

In a classical setting, once these devices are calibrated, they are forever reliable, and
the cost per use for the calibration is negligible if the device is used many times.
As quantum mechanics underpins classical mechanics, the quantumness of creating this 
shared RF must be considered by expressing the resource in quantum terms,
and this resource may be consumed because typical quantum measurements at least partially 
demolish the quantum state being measured~\cite{BRS07, BRST09}.
Studies of quantum RFs and their inherent values as resources is the domain of quantum RF theory~\cite{BRS07, BRST09, Enk05, BRS03, SVC04, GS08, GMS09, Mar11}, 
but such studies have focused mainly on pure RF resource states (except for the work in~\cite{GMS09} which focused on the relative entropy of frameness). 

Our aim here is to develop quantum RF measures that cover both pure and mixed states because pure states are an unattainable ideal.  
Many entanglement measures exist for bipartite quantum states,
but these measures are equivalent to each other for the restriction to pure bipartite quantum states.
For mixed states,  these measures can be quite different and should be understood operationally. 
Our theory exploits analogies with entanglement resource theory,      
and we concentrate solely on RF monotones that are defined over pure states and extended  to mixed states by convex-roof extensions.

Our focus will be on the case where parties lack a common reference for phase and are restricted to physical transformations that are invariant under phase changes.   
In particular,  we develop a method to construct different types of frameness monotones by
adapting existing entanglement monotones.          
As an example,  we introduce a concurrence frameness monotone. 
Also we consider a `frameness of formation' monotone   
which quantifies the number of refbits (analogous to ebits)~\cite{Enk05}
that are required to construct the resource state asymptotically
using only allowed operations and show that our frameness of formation is a proper monotone
(non-increasing on average under the set of allowed operations).
Our work creates a foundation for studies of quantum frameness for mixed-state resources 
and points to the analogies between entanglement and frameness that are ripe for exploitation.

\section{Preliminaries}
\label{FrSSR}

In this section, we reprise the basic notions of quantum RF theory, which are necessary in order to generalize from 
pure-state to mixed-state quantum RFs.   
We cast the theory in a quantum communication context in which two parties,
Alice and Bob, collaborate so that Alice can effect a completely-positive map~$\mathcal{E}$
to her state~$\rho\in \mathcal{S} (\mathscr{H})$, for~$\mathscr{H}$ the Hilbert space and~$\mathcal{S}(\mathscr{H})$ the space of  normalized states that act on $\mathscr{H}$.  Also let $\mathcal{P}(\mathscr{H})$ denote the projective Hilbert space of the Hilbert space~$\mathscr{H}$.  For a finite Hilbert space $\mathscr{H}\equiv \mathbb{C}^d$, the projective Hilbert space is the complex projective space~$\mathbb{PC}^{d-1}$. 

Alice and Bob have identical systems,  so their Hilbert spaces are  isomorphic,  
but we assume that Alice lacks the tools to perform mapping~$\mathcal{E}$ and relies on Bob,  who has this capability.
Alice sends the state~$\rho$ to Bob via a (completely-positive trace-preserving) communication channel
$\mathcal{C}:\mathcal{S} (\mathscr{H})\rightarrow \mathcal{S} (\mathscr{H})$,
but unfortunately Alice and Bob lack shared reference frame information.
We shall see that  lacking  shared reference frames can imply a superselection rule (SSR) on
the transformation~$\mathcal{E}$ that Alice, with Bob's collaboration, wishes to implement~\cite{BRS07}.

The specific channel we consider is a random unitary channel (RUC) denoted $\mathcal{U}(g)$
such that $\mathcal{U}(g)[\rho]=U(g)\rho U^\dagger(g)$
for each $U(g)$ an automorphism on~$\mathcal{H}$
and~$g\in G$ for $G$ some group parametrizing all of the possible unitary channels connecting Alice to Bob.  
Where Alice and Bob lack shared reference frame information,  $G$ is the group of transformations between the frames.  
The lack of reference frame information is manifested as a complete ignorance of~$g$ by Alice and Bob;
mathematically this complete ignorance corresponds to a uniform prior distribution for~$g$ over the Haar 
measure for the group~$G$.   

Alice sends the state~$\rho$ to Bob via the channel~$\mathcal{U}(g)$, 
and Bob then effects the mapping~$\mathcal{E}$ and sends the resultant state~$\mathcal{E}\circ \mathcal{U}(g) [\rho]$
back to Alice.    
For~$\mathcal{U}^\dagger(g)[\rho]:=U^\dagger(g)\rho U(g)$,
and given Alice's uniform lack of knowledge of~$g$,
she receives the state
\begin{equation}
\label{eq:twirl1}
	\tilde{\mathcal{E}}_G[\rho]:=\int_G\text{d}\mu(g)\;\mathcal{U}^\dagger(g)\circ\mathcal{E}\circ\mathcal{U}(g)[\rho]
\end{equation}
for~$\text{d}\mu(g)$ the group-invariant (Haar) measure.
As this relation holds for any~$\rho$,
we can write
\begin{equation}
\label{eq:twirl2}
	\tilde{\mathcal{E}}_G:=\int_G\text{d}\mu(g)\;\mathcal{U}^\dagger(g)\circ\mathcal{E}\circ\mathcal{U}(g)
	=:\mathcal{G}\left(\mathcal{E}\right), 
\end{equation}
which is known as `twirling' $\mathcal{E}$,  with~$\mathcal{G}$ being the twirling operation.
Twirling is idempotent: twirling a twirled operator leaves the twirled operator intact. 
This imposes a direct-sum structure on Alice's Hilbert space, which is a SSR~\cite{BRS07}.
We can assume that~$U$ is fully reducible
as such a representation always exists for compact Lie groups and finite Hilbert spaces.

We define a $G$-invariant operation~$\mathcal{E}$
according to $\mathcal{E}\circ\mathcal{U}(g)=\mathcal{U}(g)\circ\mathcal{E}$
for all $g\in G$.
Thus, $\tilde{\mathcal{E}}_G=\mathcal{E}$ if and only if~$\mathcal{E}$ is a $G$-invariant operation.
Therefore the lack of reference information is not an impediment for Alice and Bob to collaborate to effect~$\mathcal{E}$
if~$\mathcal{E}$ is $G$-invariant.

Now we consider the U(1)-invariance case following the approach of Gour and Spekkens~\cite{GS08}.
The abelian group U(1) has a unitary representation $\theta\mapsto\exp(-\imath \theta \hat{n})$
for~$\hat{n}$ the Hermitian generator with spectrum $\{n\in\mathbb{N}\}$.
We refer to~$\hat{n}$ as a number operator.

The Hilbert space can be expressed as the direct sum~$\mathscr{H}=\bigoplus_n\mathscr{H}_n$
with~$n$ an irrep index for U(1) and~$\mathscr{H}_n$ the multiplicity subspaces.  The eigenstates $|n,\beta\rangle$ of the number operator form a  basis for~$\mathscr{H}_n$ where $\beta$ is a multiplicity index.   Operations on multiplicity spaces are unaffected by the U(1)-SSR and as a result
any pure state can be transformed via the U(1)-invariant unitary transformation to a standard form.  Consider the pure state 
\begin{equation} 
 |\psi\rangle=\sum_{n}c_n |\psi_n \rangle 
\end{equation}
where $|\psi_n\rangle \in \mathscr{H}_n$  are normalized states, and let~$\lambda_n:=|c_n|^2$. We apply the Gram-Schmidt process to extend $|\psi_n\rangle$ to  a full orthonormal  basis $\left\{ |\psi_n\rangle\right\}\cup\left\{ |\phi_{n,\beta}\rangle\right\}_{\beta}$ of the subspace $\mathscr{H}_n$.  The unitary transformation 
\begin{align}
U:=\sum_n \left(\frac{c_n^*}{|c_n|} |n,0\rangle\langle\psi_n|+\sum_{\beta\neq0} |n,\beta\rangle\langle\phi_{n,\beta}|\right)
\end{align}
is U(1)-invariant and takes the state $|\psi\rangle$  to the standard form 
\begin{align}
\label{psi}
|\psi\rangle=\sum_n\sqrt{\lambda_n} |n\rangle  
\end{align}
where $|n\rangle$ is shorthand notation for the fixed choice of the state $|n,0\rangle$.    
The spectrum of the state~$|\psi\rangle$  
 is defined to be the set 
 \begin{align}
 \label{Spec}
 \text{spec}(|\psi\rangle):=\{n;\lambda_n>0\}.     
 \end{align}

The~$G$-invariant operator~$\mathcal{E}$ can be expressed as a Kraus operator decomposition
\begin{equation}
	\mathcal{E}[\rho]=\sum_\ell\hat{K}^{(\alpha)}_\ell\rho\hat{K}^{(\alpha)\dagger}_\ell, 
\end{equation}
and for general U(1)-invariant CP maps on standard states,  
Kraus operators must have the form 
\begin{equation}
\label{U(1)Kraus}
	\hat{K}^{(\alpha)}_\ell=\sum_n k^{(\alpha)}_{\ell,n} |n+\ell\rangle\langle n|
\end{equation}
for~$\ell$ an integer and~$k^{(\alpha)}_{\ell,n}\in\mathbb{C}$
such that $\sum_i |k^{(\alpha)}_{\ell,n}|^2 \le 1$ with equality holding if the transformation is trace-preserving~\cite{GS08}.
In this notation, $\ell$ represents the number-shift imposed by the Kraus operator
and~$\alpha$ an index for a particular $\ell$-shifting Kraus decomposition.

States that are not $G$-invariant are resources that Alice or Bob can use to circumvent SSR restrictions,
and ``frameness'' denotes this quantum resource.   
Here we focus on the~U(1)-SSR that corresponds to lacking a common phase,
for example the phase of a laser in homodyne measurements or orientation in a plane.

Note that we express everything with respect to Alice's RF and make a distinction between the preparation procedure by Alice and the consequent transformations of the prepared state performed by Bob,  who has access to the prepared state only through the twirling channel.  Alice can prepare any state  including coherent superpositions that are restricted by the SSR.  However, as all operations afterwards are performed by Bob who does not have access to Alice's RF, the transformations of the state have to be $G$-invariant. Thus,  a coherent superposition, like the state $|\psi\rangle$ in Eq.~(\ref{psi}),  is distinct from the mixture 
\begin{equation}
\rho:= \sum_n \lambda_n |n\rangle\langle n|
\end{equation}
that results from twirling the state. Let us compare the case where Alice prepares the coherent state $|\psi\rangle$ versus the case where she prepares the invariant state $\rho$. Of course,  Bob receives the twirled state $\rho$ in either case and he is free to perform any operation on the state he receives relative to his own RF before sending the state back to Alice. With respect to Alice's RF, however,  the  net result is a U(1)-invariant transformation on a coherent supersposition in the first case and on the twirled mixture in the second case~(See Eqs.~(\ref{eq:twirl1}) and ~(\ref{eq:twirl2})). The two cases  are distinct. For example,  $|\psi\rangle$ can be transformed to $\rho$  while $\rho$ cannot be transformed to $|\psi\rangle$ by U(1)-invariant operations, ~i.\ e.\ by Bob when viewed in Alice's RF.  A state like $|\psi\rangle$ that is not $G$-invariant is a resource, while $G$-invariant states like $\rho$ are not. Alice can accompany a resource state (known to herself and to Bob) with the target state that Bob is supposed to act on and send them together to Bob. This way, the resource state acts as a token of Alice's RF and can be used to partially overcome the SSR-restriction on transformations. In other words, there exist joint $G$-invariant transformations on the two states whose net effect on the target state alone is equivalent to a transformation that is  no longer $G$-invariant~\cite{BRS07, BRST09}. If Alice accompanies a $G$-invariant state instead of a resource with the target state, this is no longer possible. $G$-invariant states are non-resource states.

Frameness monotones are functions that measure the strength of resource states. We now begin with a discussion of frameness monotones in general.

\section{Frameness Monotones} 
\label{mono}

A quantum RF resource should be a monotone in order that the amount of frameness does not
increase under allowed operations.

The motivation for monotones is that they provide operational measures to quantify the strength of  resources.   If one is faced with a certain task and one wants to know what is the maximum probability with which the task can be  performed or the maximum number of particular states  that one can acquire using $G$-invariant operations,  
the definition of that task already involves an optimization over all allowed operations. It shouldn't be possible to do better by first pre-processing the resource because the definition of the task  assumes that all the preprocessing has already been done. In such cases, any measure of the strength of the resource cannot increase under the restricted operations.   Entanglement monotones,  which are attempts to measure entanglement, are functions over states that  are non-increasing under local operations and classical communication (LOCC). Here we consider functions that are non-increasing under $G$-invariant operations~\cite{GS08, GMS09}. 

Monotones can also be used to  determine whether certain states can or cannot be transformed to each other under the SSR-restrictions. A state with a lower value of a frameness monotone cannot be turned into a state with higher value of the monotone using group symmetric transformations. Thus,  monotones under $G$-invariant operations  are also studied in the context of symmetric dynamics. Monotones quantify the asymmetry of quantum states, and provide new conditions beside those specified by Noether's theorems and the related conservation laws to determine the consequences of symmetry for mixed states in closed dynamics as well as pure and mixed states in dissipative and open systems.   Frameness monotones are  known as asymmetry monotones in this context~\cite{Mar11}.

In this section,  we establish a set of reasonable conditions that a valid frameness measure should satisfy,
and we provide insight and background for this choice of conditions. 

The transformation~$\mathcal{E}$ can be decomposed into a set of completely positive operators~$\{\mathcal{E}_x\}$
with $\mathcal{E}=\sum_x\mathcal{E}_x$.
For input state~$\rho$,
the output state is expressed as the unit-trace state
\begin{equation}
\label{eq:sigmax}
	\sigma_x:=\mathcal{E}_x[\rho]/p_x,\;
	p_x:=\text{Tr}\left(\mathcal{E}_x[\rho]\right)
\end{equation}
with~$p_x$ the probability of the~$x^\text{th}$ outcome.
The ensemble of outcomes is written as $\left\{ \sigma_x, p_x\right\}$.

\begin{definition}
\label{def:ensemble}
A function $F: \mathcal{S} \left(\mathscr{H}\right)\to\mathbb{R}^{+}$ is an \emph{ensemble frameness monotone}
if it satisfies
\begin{enumerate}
\item [F1.]  $F(\rho)=0$ for any $G$-invariant state~$\rho=\mathcal{G} (\rho)$;
\item [F2.] $F(\rho)\geq\sum_x p_x F(\sigma_x)$, for each~$\mathcal{E}_x$ being $G$-invariant;
\item [F3.] $F$ is convex: for any ensemble~$\left\{ \sigma_i, p_i\right\}$,
$\sum_i p_i F(\sigma_i)\ge F\left(\sum_i p_i  \sigma_i\right)$.
\end{enumerate}
\end{definition}

Conditions F1-F3 are analogous  to Vidal's criteria for entanglement monotones~\cite{Vid00}.
In general Condition~F3 may not be necessary as logarithmic negativity (which provides an upper bound for distillable entanglement) is a useful measure of entanglement although not convex~\cite{Ple05}. 
However, here we restrict to convex measures and therefore require all three Conditions F1-F3.    

 We can also define pure-state frameness monotones in a similar manner.  Pure-state frameness monotones are functions that behave monotonically under the more restricted set of $G$-invariant operations that map pure states to pure states only.

\begin{definition}
\label{def:ensemble-pure}
A function 
\begin{equation}
F_\text{pure}: \mathcal{P} (\mathscr{H}) \rightarrow \mathbb{R}^+ :|\psi\rangle\langle\psi|\mapsto F_\text{pure}\left(|\psi\rangle\langle\psi|\right)\nonumber
\end{equation}
is an ensemble pure-state frameness monotone if it satisfies 
\begin{enumerate}
\item [F1.]  $F_\text{pure}(|\psi\rangle\langle\psi|)=0$ for any $G$-invariant state~$|\psi\rangle\langle\psi|=\mathcal{G} (|\psi\rangle\langle\psi|)$;
\item [F2.] $F_\text{ pure}(|\psi\rangle\langle\psi|)\geq\sum_x p_x F_\text{pure} (|\phi_x\rangle\langle\phi_x|)$,  for $G$-invariant transformations ~$\mathcal{E}_x$,  such that 
$
|\phi_x\rangle\langle\phi_x|:=\mathcal{E}_x[|\psi\rangle\langle\psi|]/p_x,\;
	p_x:=\text{Tr}\left(\mathcal{E}_x[|\psi\rangle\langle\psi|]\right). 
$
\end{enumerate}
Here  $\mathcal{P}(\mathscr{H})$ denote the projective Hilbert space of the Hilbert space~$\mathscr{H}$. 
\end{definition}

If an ensemble frameness monotone already exists for  pure states, one way to extend the pure-state monotone to a  measure defined for all states~$\rho$ is according to the following definition:
\begin{definition}
 Given a pure-state frameness monotone
\begin{equation}
	F_\text{pure}:  \mathcal{P} (\mathscr{H})\rightarrow \mathbb{R}^+:|\psi\rangle\langle\psi|\mapsto F_\text{pure}\left(|\psi\rangle\langle\psi|\right), \nonumber
\end{equation}
 the convex-roof extension $F$ is defined by 
\begin{align}
\label{eq:convexroof}
	F:\mathcal{S}(\mathscr{H})&\rightarrow \mathbb{R}^+:\\ \nonumber \rho&\mapsto F(\rho)\equiv \min_{\left\{|\psi_i\rangle, p_i\right\}} \sum_i p_i \; F_\text{pure} \left(|\psi_i\rangle\langle\psi_i|\right),
\end{align}
with the minimum taken over all possible pure-state decompositions of  $\rho=\sum_ip_i|\psi_i\rangle\langle\psi_i|$.    
\end{definition}
For pure states (i.\ e.\ ~rank-one density operators) the two monotones are of course always equal:
\begin{align}
F(|\psi\rangle\langle\psi|)\equiv F_\text{pure}(|\psi\rangle\langle\psi|), \;\;\; \forall |\psi\rangle \in \mathscr{H}. 
\end{align} 

As we presently show the convex-roof extension of a pure-state ensemble monotone is an ensemble monotone for all states based on Definition~\ref{def:ensemble}.    
To see this, consider the following two lemmas that follow directly from the definition of the convex-roof extension.  

\begin{lemma}
\label{lemma:lm1}
The convex-roof extension of a pure-state frameness monotone is a convex function. 
\end{lemma}
\begin{proof}
Let   $\rho=\sum_i p_i\rho_i$, and let $\rho_i=\sum_j p_{ij} |\psi_{ij}\rangle \langle \psi_{ij}|$ be the optimal decomposition of $\rho_i$ in the sense of Eq.~(\ref{eq:convexroof}).
The minimum average frameness of~$\rho$ is reached either by the sum 
\begin{equation}
	F(\rho)=\sum_{i,j} p_i p_{ij} F_\text{pure}\left(|\psi_{ij}\rangle\langle\psi_{ij}|\right)=\sum_i p_i F(\rho_i)
\end{equation}
or by some other ensemble~$\left\{|\phi^{(\alpha)}_\ell\rangle, q_\ell \right\}$ forming~$\rho$, in which case
\begin{equation}
	F(\rho)= \sum_\ell q_\ell F_\text{pure}\left(|\phi_\ell\rangle\langle\phi_\ell|\right)< \sum_i p_i F(\rho_i)
\end{equation}  
so that in general $F(\rho)\le \sum_i p_i F(\rho_i)$.
\end{proof}
\begin{lemma}
\label{lemma:lm2}
If 
\begin{equation}
F_\text{pure}:  \mathcal{P} (\mathscr{H}) \rightarrow \mathbb{R}^+ :|\psi\rangle\langle\psi|\mapsto F_\text{pure}\left(|\psi\rangle\langle\psi|\right)  \nonumber  
\end{equation}
  does not increase on average under $G$-invariant transformations  
  between pure states  (i.\ e.\  $F_\text{pure}$ is a pure-state ensemble monotone) 
then the convex-roof extension defined by Eq.~(\ref{eq:convexroof}) is an ensemble frameness monotone. 
\end{lemma}
\begin{proof}

We need to show that,
for any~$\rho$ and~$G$-invariant operation $\mathcal{E}=\sum_x\mathcal{E}_x$,
we have~$F(\rho) \ge \sum_xp_xF(\sigma_x)$, where
\begin{equation}
\label{eq:sigmax2}
	\sigma_x:=\mathcal{E}_x[\rho]/p_x,\;
	p_x:=\text{Tr}\left(\mathcal{E}_x[\rho]\right)
\end{equation}
with~$p_x$ the probability of the~$x^\text{th}$ outcome.
The ensemble of outcomes is written as $\left\{ \sigma_x, p_x\right\}$.
 Assume~$\left\{|\psi_i\rangle, q_i\right\}$    
is the optimal decomposition of~$\rho$ in the sense that
$F(\rho)=\sum_iq_iF_\text{pure}(|\psi_i\rangle\langle\psi_i|)$.   Let $\hat{K}^{(\alpha)}_{x,\ell}$ be a choice of Kraus operators for $\mathcal{E}_x$.   
Each $\hat{K}^{(\alpha)}_{x,\ell}$ effects the mapping
\begin{equation}
		 |\psi_i\rangle\mapsto |\phi^{(\alpha)}_{x,i,\ell}\rangle
		:= \frac{\hat{K}^{(\alpha)}_{x,\ell}}{\sqrt{q^{(\alpha)}_{x,i,\ell}}}\left|\psi_i\right\rangle
\end{equation}   
with  probability  
$q^{(\alpha)}_{x,i,\ell}=\|\hat{K}^{(\alpha)}_{x,\ell} |\psi_i\rangle\|^2$. 
Thus,
\begin{equation}
	\sigma_x=\frac{1}{p_x}\sum_{i.\ell,\alpha}q_iq^{(\alpha)}_{x,i,\ell} |\phi^{(\alpha)}_{x,i,\ell}\rangle\langle \phi^{(\alpha)}_{x,i,\ell}|.
\end{equation}

The convex-roof extension is a convex function (Lemma~\ref{lemma:lm1})  so  
\begin{align}
	F(\sigma_x)&\le \frac{1}{p_x}\sum_{i.\ell,\alpha}q_iq^{(\alpha)}_{x,i,\ell}
	F(|\phi^{(\alpha)}_{x,i,\ell}\rangle\langle\phi^{(\alpha)}_{x,i,\ell}|)\nonumber\\
	 &=\frac{1}{p_x}\sum_{i.\ell,\alpha}q_iq^{(\alpha)}_{x,i,\ell} 
	F_\text{pure}(|\phi^{(\alpha)}_{x,i,\ell}\rangle\langle\phi^{(\alpha)}_{x,i,\ell}|)
\end{align}
The operators $\hat{K}^{(\alpha)}_{x,\ell}$ are themselves $G$-invariant, and as we have assumed that $F_\text{ pure}$ is an ensemble monotone on pure states,
\begin{equation}
	         F_\text{pure}\left(|\psi_i\rangle\langle \psi_i|\right)\ge \sum_{x,\ell,\alpha} q^{(\alpha)}_{x,i,\ell}
		F_\text{pure}\left(|\phi^{(\alpha)}_{x,i,\ell}\rangle\langle\phi^{(\alpha)}_{x,i,\ell}|\right)
\end{equation}
readily follows.
Putting everything together, we obtain 
\begin{align}
	F(\rho)
		=&\sum_i q_i \; F_\text{pure}\left(|\psi_i\rangle\langle \psi_i|\right)\nonumber\\
		\ge& \sum_{i,x,\ell,\alpha} q_i q^{(\alpha)}_{x,i,\ell}
			F_\text{pure}\left(|\phi^{(\alpha)}_{x,i,\ell}\rangle\langle\phi^{(\alpha)}_{x,i,\ell}|\right)\nonumber\\
		\ge&\sum_xp_x F(\sigma_x)
\end{align}
so the convex-roof extension is indeed an ensemble monotone.
\end{proof}

In other words, if a pure-state frameness measure is an ensemble monotone on pure states,
then the convex-roof extension is an ensemble monotone under all allowed CP-maps.
Therefore, we need only consider how a function behaves on pure-state to pure-state transformations
in order to determine if it is an ensemble monotone.   
Also note that we do not interpret the convex-roof extension in terms of cost of forming the state at this stage.   
Rather, we treat the convex-roof extension only as an ensemble frameness monotone under $G$-invariant transformations.

In the next section we present the main result of our paper. The main result of the paper consists of a general method to construct pure-state frameness monotones for phase-invariant channels that we then extend  to frameness monotones for all states via the convex-roof extension.  We then narrow our study to  specific measures that quantify the frameness of bounded-size 
quantum RFs for phase.

\section{Qudit monotones for U(1)-SSR}
\label{qudit-mono}  
We are now in a position to construct useful frameness monotones for $d$-dimensional qudits.
Although we discuss quantum RFs for phase, 
our results apply to abelian symmetry groups in general.         

We bring all states into the standard form without multiplicities:  $|\psi\rangle=\sum_n\sqrt{\lambda_n}|n\rangle$.   
Let $n_\text{min}(|\psi\rangle)$ and~$n_\text{max}(|\psi\rangle)$ denote the minimum and maximum 
values of~$n$ in the number spectrum defined in Eq.~(\ref{Spec}),  with the restriction
\begin{equation}
\label{eq:nrestriction}
	n_\text{max}(|\psi\rangle)-n_\text{min}(|\psi\rangle)\le d.
\end{equation}
Consider a U(1)-invariant transformation  $\hat{K}^{(\alpha)}_\ell$ of Eq.~(\ref{U(1)Kraus}) that maps
\begin{equation}
	|\psi\rangle\mapsto|\phi^{(\alpha)}_\ell\rangle \equiv \frac{\hat{K}^{(\alpha)}_\ell}{\sqrt{p^{(\alpha)}_\ell}}\left|\psi\right\rangle
\end{equation}
with probability $p^{(\alpha)}_\ell=\|\hat{K}^{(\alpha)}_\ell |\psi\rangle\|^2$, where $\hat{K}^{(\alpha)}_\ell$ effects the mapping
\begin{equation}
	 |\psi\rangle\mapsto\ \sum_n \sqrt{\lambda_n} k^{(\alpha)}_{\ell,n} |n+\ell\rangle. 
\end{equation} 
Let us define
\begin{equation}
\label{eq:rhopsi}
	\rho_\psi:=\mathcal{G}\left(|\psi\rangle\langle\psi|\right)
\end{equation}
and  purify $\rho_\psi$ by adding an auxiliary reference system~R to obtain 
\begin{equation}
\label{Psitilde}	|\tilde{\psi}\rangle
		\equiv\sum_n \sqrt{\lambda_n} |n \rangle_\text{S}
		\otimes\left|n_\text{max}(|\psi\rangle)-n\right\rangle_\text{R} \in\mathscr{H}_\text{S}\otimes\mathscr{H}_\text{R}
\end{equation} 
with~S signifying the original system.
If we follow the same procedure for the state~$|\phi^{(\alpha)}_\ell\rangle$,
noting that  $n_{\max}(|\phi^{(\alpha)}_\ell\rangle)=n_\text{max}(|\psi\rangle)+\ell$, we use 
\begin{align}
\label{Phitilde}
	|\tilde{\phi}^{(\alpha)}_\ell\rangle
		:=&\frac{1}{\sqrt{p^{(\alpha)}_\ell}}\sum_n \sqrt{\lambda_n} k^{(\alpha)}_{\ell,n} |n
			+\ell \rangle_\text{S}\nonumber\\&
			\otimes|n_{\max}(|\phi^{(\alpha)}_\ell\rangle)-(n+\ell)\rangle_\text{R} \\
		=&\frac{1}{\sqrt{p^{(\alpha)}_\ell}}\sum_n \sqrt{\lambda_n} k^{(\alpha)}_{\ell,n} |n
			+\ell \rangle_\text{S} \otimes |n_\text{max}(|\psi\rangle)-n\rangle_\text{R}.\nonumber
\end{align} 
Evidently,  $|\tilde{\psi}\rangle$ can be transformed to $|\tilde{\phi}^{(\alpha)}_\ell\rangle$ via the local transformation $\hat{K}^{(\alpha)}_\ell\otimes \openone_\text{R}$, where~$\openone_\text{R}$ is the identity operator on~R.

The states $|\tilde{\psi}\rangle$ and $|\tilde{\phi}^{(\alpha)}_\ell\rangle$
are dependent in the sense that one state can map to the other via operations acting only on the system~S,
i.e.\ local operations.  
The fact that the two purifications in Eqs.~(\ref{Psitilde}) and~(\ref{Phitilde}) can be linked together by a local operation is due to the SSR restriction on the operations $\hat{K}^{(\alpha)}_\ell$. To see this, suppose  the restriction was lifted to allow a number state~$|n\rangle$ to transform to a superposition of two number states
$|n_1\rangle+|n_2\rangle$ with $n_2> n_1$.
The purification process in~(\ref{Psitilde}) maps the outcome superposition to an entangled state
\begin{equation}
	|\tilde{\phi}^{(\alpha)}_\ell\rangle\propto \sqrt{\eta_1}|n_1\rangle_\text{S} \otimes |n_2
		-n_1\rangle_\text{R}+\sqrt{\eta_2} |n_2\rangle_\text{S}\otimes |0\rangle_\text{R}, 
\end{equation}
 whereas the purified version of the initial state, $|\tilde{\psi}\rangle= |n\rangle_\text{S}\otimes|0\rangle_\text{R}$, 
is separable, and no local operation can make it entangled. 

Our particular choice of auxiliary states in the purification process ensures that purified states always remain within a superselected block of some fixed total number as for Eq.~(\ref{Psitilde}).
Thus,  we can see that the  operators  $\hat{K}^{(\alpha)}_\ell$ in~(\ref{U(1)Kraus})  are precisely those Kraus operators that act on system S alone and, at the same time,  either keep the joint state of the two systems  S and R within the multiplicity space of total number  $n_\text{max}(|\psi\rangle)$ or transfer them both to the multiplicity space
of another total number $n_\text{max}(|\psi\rangle)+\ell$ for some~$\ell\ge -n_\text{max}(|\psi\rangle)$. System R acts as a sort of quantum phase reference in the sense that it enables system S to break the SSR locally while preserving the overall SSR. The partial trace that results from lack of access to the reference system R  is equivalent to the twirling map on the initial unipartite state. 

All bipartite pure  state entanglement monotones can be expressed as concave functions of the Schmidt coefficients of the states, or equivalently, the eigenvalues of the reduced density matrix~\cite{Vid00}. 
The reduced density matrix of the state  $|\tilde{\psi}\rangle$ is the same as $\rho_\psi$ defined in Eq.~(\ref{eq:rhopsi}).
Thus, from any entanglement monotone function defined  for states acting on~$\mathscr{H}_\text{S} \otimes \mathscr{H}_\text{R}$,  we can build a monotone under U(1)-SSR for states acting on~$\mathscr{H}$ by replacing the partial trace with the twirling map.     
We formalize this result in the following proposition:     
\begin{proposition}
\label{Vidalf}
Suppose a function $f:  \mathcal{S} (\mathscr{H}) \rightarrow \mathbb{R}^+$ satisfies the following two conditions.
\begin{enumerate}
\item [E1.] Unitary invariance:
	$f(\rho)=f(U \rho \;U^{\dagger})\:\forall U\in \mathcal{S} (\mathscr{H})$.
\item [E2.] Concavity: $f(t \rho_1+[1-t]\rho_2)\ge t f(\rho_1)+[1-t]f(\rho_2)\;\forall\;t \in [0,1]$.
\end{enumerate} 
Then
\begin{equation}
\label{eq:FfG}
	F_\text{pure}: \mathcal{P} \left(\mathscr{H}\right)\rightarrow\mathbb{R}^+:
		\left|\psi\right\rangle\left\langle\psi\right|\mapsto f(\mathcal{G}\left(|\psi\rangle\langle \psi|\right)),\nonumber
\end{equation}    
is a pure-state ensemble monotone under U(1)-invariant operations, and its convex-roof extension 
\begin{equation}
F: \mathcal{S}\left(\mathscr{H}\right) \rightarrow \mathbb{R}^+,
\end{equation} 
is an ensemble monotone for all states.  
\end{proposition}

\begin{proof}
We first prove that $F_\text{pure}$ is a pure-state ensemble monotone. 
Let $\mathcal{E}=\sum_x\mathcal{E}_x$ for~$\mathcal{E}_x$ being U(1)-invariant completely positive operators that map pure states to pure states.    
Given a pure state~$|\psi\rangle$, we follow the notation of Definition~\ref{def:ensemble-pure}: 
\begin{equation}
	|\phi_x\rangle\langle \phi_x|:=\frac{1}{p_x}\mathcal{E}_x\left[|\psi\rangle\langle\psi|\right],\;
	p_x:=\text{Tr}\left(\mathcal{E}_x\left[|\psi\rangle\langle\psi|\right]\right).
\end{equation}
Let the corresponding U(1)-invariant Kraus operators be indexed as 
\begin{equation}
	\hat{K}_{x,j} |\psi\rangle = \sqrt{p_{x,j}} |\phi_{x}\rangle, 
\end{equation}
where $\sum_j p_{x,j}\equiv p_x$.  
The purifications of $\rho_\psi$ and~$\rho_{\phi_{x}}$ according to Eqs.~(\ref{Psitilde}) and~(\ref{Phitilde}) are then related to each other by
\begin{equation}
 	\hat{K}_{x,j}\otimes \openone_R |\tilde{\psi}\rangle = \sqrt{p_{x,j}} |\tilde{\phi}_{x}\rangle.
\end{equation}

The reduced density operator of the auxiliary system~$R$ does not change under the transformation.    
For   
\begin{equation}
	\tau:=\text{Tr}_\text{S}\left(|\tilde{\psi}\rangle\langle\tilde{\psi}|\right),\;
	\tau_{x}\equiv\text{Tr}_\text{S}\left(|\tilde{\phi}_{x}\rangle\langle\tilde{\phi}_{x}|\right),
\end{equation}
we obtain
\begin{equation}
	\tau=\sum_{x,j} p_{x,j}\tau_{x}=\sum_{x} p_{x} \tau_{x}.
\end{equation}
Condition~E1 (unitary invariance) ensures that~$f$ is a function only of the state's eigenvalues,
and concavity ensures that 
\begin{equation}
\label{Reduced}
	f(\tau)=f\left(\sum_{x}p_{x}\tau_{x}\right)\ge\sum_{x}p_x f\left(\tau_{x}\right).
\end{equation}
On the other hand
\begin{align}
\label{Traces}
	f\left(\text{Tr}_\text{S}\left[|\tilde{\psi}\rangle\langle\tilde{\psi}|\right]\right)
		=&f\left(\text{Tr}_\text{R}\left[|\tilde{\psi}\rangle\langle\tilde{\psi}|\right]\right)\nonumber\\
		=&f\left(\mathcal{G}\left(|\psi\rangle\langle\psi|\right)\right)
		=F_\text{pure}\left(|\psi\rangle\langle\psi|\right)
\end{align}
and similarly for $\left\{|\phi_{x}\rangle\right\}$.
Finally,  Eqs.~(\ref{Reduced}) and~(\ref{Traces}) together imply
\begin{align}
	F_\text{pure}\left(\left|\psi\right\rangle\left\langle\psi\right|\right)\ge&\sum_{x} p_x
		F_\text{pure}\left(\left|\phi_{x}\right\rangle\left\langle\phi_{x}\right|\right),   
\end{align}
which is the desired result.  
Lemma~\ref{lemma:lm2} ensures that the convex-roof extension $F$  defined by Eq.~(\ref{eq:convexroof})  is also an ensemble monotone for all states.   
\end{proof}

 Note that     
$
F_\text{pure}(|\psi\rangle\langle\psi|):=f(\rho_\psi)
$ 
in Eq.~(\ref{eq:FfG}) is an ensemble monotone function of the pure state~$|\psi\rangle$ and not the mixed state~$\rho_\psi$ (Eq.~(\ref{eq:rhopsi})). Also the convex-roof extension~$F$ is no longer  the same function as $f$, although both $F$ and $f$ act on the whole state space $\mathcal{S}(\mathscr{H})$. In fact,  the function $f$ by itself is  not a monotone at all. For example,  no condition is set for the behavior of $f$ under general $U(1)$-invariant CP-maps. Furthermore, by Condition E2, $f$ is concave under mixtures of two or more states whereas Condition F3 of Definition~\ref{def:ensemble} states that a frameness monotone has to be a convex function.    

We can now build the counterparts of Vidal's entanglement monotones for pure states~\cite{Vid00}.
Let 
\begin{equation}
	\bm{\lambda}^{\downarrow} (\rho_\psi)=\left(\lambda_1^{\downarrow}, \ldots, \lambda_{d}^{\downarrow}\right)
\end{equation}
be the vector obtained  by rearranging the coordinates of $\bm{\lambda} (\rho_\psi)$ in decreasing order.    
\begin{corollary}
\label{cor:U1family}
The family of pure-state functions
\begin{equation}
\label{Vmonotones}
	F_k:  \mathcal{P} (\mathscr{H}) \rightarrow \mathbb{R}^+:
		|\psi\rangle\langle\psi|\mapsto\sum_{i=k}^d \lambda_i^{\downarrow},\;
		k=2, \ldots, d
\end{equation}
together with their convex-roof extensions
are a family of  U(1)-frameness ensemble monotones. 
\end{corollary}
This family of functions clearly satisfies both conditions of Prop.~\ref{Vidalf}.
Consequently,  we see that the outcomes of any U(1)-invariant transformation majorize the initial state on average.  This  generalizes  what was already established for the case of deterministic transformations~\cite{GS08}.  

As another example, consider the entropy of the twirled state that is both concave and unitarily invariant.  
\begin{corollary}
\label{cor:Entropy}
The entropy of frameness 
\begin{align}
S_F(|\psi\rangle\langle \psi|):=-\rho_\psi \log_2 \rho_\psi
\end{align} 
is an ensemble monotone under a U(1)-SSR. 
\end{corollary}
The entropy of frameness is equal to the relative entropy of frameness (G-asymmetry) for pure states~\cite{GMS09}. 

Conditions 1 and 2 in Prop.~\ref{Vidalf} are only sufficient conditions for U(1)-monotones and not necessary ones,
not even for monotones defined over pure states. If $|\tilde{\psi}\rangle$ can be transformed to $|\tilde{\phi}^{(\alpha)}_\ell\rangle$ under general local operations,
it does not follow that $|\psi\rangle$ is necessarily transformable into $|\phi\rangle$ under a U(1)-SSR.
This reasoning follows simply because the local transformation that takes $|\tilde{\psi}\rangle$ to  $|\tilde{\phi}^{(\alpha)}_\ell\rangle$ need not have Kraus operators of the form $\hat{K}^{(\alpha)}_\ell\otimes \openone_R$
for~$\hat{K}^{(\alpha)}_\ell$ specified in Eq.~(\ref{U(1)Kraus}).
Thus, the frameness monotones, unlike entanglement monotones,   do not have to remain non-increasing on average for all local operations and therefore need not be of the form derived in Prop.~\ref{Vidalf}.  

As a counterexample, consider the normalized number variance 
\begin{equation}
\label{eq:Var}
	V_\text{pure} (|\psi\rangle \langle \psi| )=4\left(\langle\psi|\hat{n}^2|\psi \rangle- \left\langle\psi \right|\hat{n}\left| \psi \right>^2\right).  
\end{equation}  
The variance is neither concave nor convex, and yet  it was shown to be an ensemble monotone over pure states~\cite{SVC04, GS08}.  
For similar reasons, majorization is a necessary but not a sufficient condition for pure-state to pure-state deterministic transformations.
Thus,  the  U(1)-frameness monotones of Eq.~(\ref{Vmonotones}), unlike Vidal's monotones in entanglement theory,  do not  fully characterize deterministic U(1)-invariant transformations~\cite{Vid00}.  

Motivated by Wootters's formula for the  concurrence of bipartite two-qubit states~\cite{Woo98}, later extended to bipartite qudit states~\cite{RBC+01, Gou05a, Gou05b}, we can also construct a family of concurrence measures for qudits with $d\ge 2$.       
Let  
\begin{equation}
\label{eq:sk}
	S_k( \bm{\lambda}(\rho_\psi))\equiv \sum_{m_1< m_2<\cdots< m_k} \lambda_{m_1} \lambda_{m_2}\ldots\lambda_{m_k},
\end{equation}
for $k=2,\ldots,d$,  
denote the~$k^{\text{th}}$ elementary symmetric function  of the eigenvalues $\bm{\lambda}(\rho_\psi)=\left(\lambda_1, \ldots, \lambda_d\right) $ of $\rho_\psi$.   
We assume that  $\lambda_n\equiv 0$ for $n_\text{max}(|\psi\rangle)<n\le d$.   
\begin{definition}
The family of concurrence-of-frameness functions are defined for  pure states  as   
\begin{align}
\label{eq:fk}
	C_k\left(\left|\psi\right\rangle\langle\psi|\right)
		:& \mathcal{P} (\mathscr{H}) \rightarrow \mathbb{R}^+\nonumber\\
		:&\left|\psi\right\rangle\left\langle\psi\right|\rightarrow
		f_k(\rho_\psi)
				:=\left[\frac{S_k\left(\bm{\lambda}(\rho_\psi)\right)}
					{S_k\left(\frac{1}{d},\ldots,\frac{1}{d}\right)}\right]^{\frac{1}{k}}
\end{align}
and extended to mixed states via their convex-roof extensions. 
\label{cf}
\end{definition}
The~$f_k$ are concave functions of $\bm{\lambda}(\rho)$~\cite{Gou06}.
Hence Prop.~\ref{Vidalf} guarantees that $\{C_k\}$ are ensemble monotones as summarized in the following corollary. 
\begin{corollary}
The concurrence~$C_k(\rho)$ for $k=2, \ldots, d$  of a state~$\rho$  does not increase on average under U$(1)$-invariant operations.    
\label{cf1}
\end{corollary}
Note that non-resource states are number eigenstates hence not of full rank.
Thus, their concurrence is identically zero as is expected  from Condition F1 in Def.~\ref{def:ensemble}.    


Now we demonstrate the similarity between the entanglement and frameness resource measures by calculating the concurrence of mixed qubit states.
For a pure single-qubit state ($d=2$), 
$C_2(|\psi\rangle)=|\langle\psi|X|\psi^{*}\rangle|$
for  complex conjugation in the basis $\left\{|0\rangle,|1\rangle\right\}$
and~$X=|0\rangle\langle1|+|1\rangle\langle0|$ the flip operator in this basis.
Let  
\begin{equation}
\label{R}
	R:=\sqrt{\sqrt{\rho}\tilde{\rho}\sqrt{\rho}},\;
	\tilde{\rho}:= X\rho^{*}X,
\end{equation}
and let the set of eigenvalues of $R$ be 
$\mu(R)=\left\{\mu_1, \mu_2\right\}$.
In Appendix~\ref{sec:Ap1}, we derive the explicit dependence of $\mu_1$ and~$\mu_2$  on the parameters of the spectral decomposition of~$\rho$.

\begin{proposition}
The concurrence of frameness for a qubit state~$\rho$ is 
\begin{equation}
	C_2\left(\rho\right)=| \mu_1-\mu_2|.  \label{lambda}
\end{equation}
\end{proposition}
\begin{proof}
The proof is similar to Wootters's proof for concurrence of entanglement~\cite{Woo98}.
Without loss of generality  we assume that $\mu_1\geq\mu_2$.
Let
\begin{equation}
	\rho=  |\phi_1\rangle\langle\phi_1| + |\phi_2\rangle\langle\phi_2| 
\end{equation}
be the spectral decomposition of~$\rho$, where each state~$|\phi_i\rangle$ is unnormalized.
Also  let  
$\tau_{ij}:= \langle \phi_i|\tilde{\phi}_j\rangle$,
which yields the symmetric relation $\tau_{ij}=\tau_{ji}$~\cite{Woo98}.
Then
 \begin{equation}
	 R^2=\sum_{i,j} (\tau\tau^{*})_{ij} |\phi_i\rangle \langle \phi_j|, 
 \end{equation}
and the unitary operator $U$ that diagonalizes $R^2$ relates the spectral decomposition of~$\rho$ to a different  decomposition 
 $\rho= |\xi_1\rangle\langle \xi_1|+ |\xi_2\rangle\langle\xi_2|$   
  for which the corresponding matrix $\tau'_{ij}:=\langle\xi_i|\tilde{\xi}_j\rangle$ is diagonal,  and 
\begin{equation}
	\tau'_{11}= \mu_1,\; \tau'_{22}= -\mu_2\;,
	\tau'_{12}=\tau'_{21}=0.
\end{equation}

The average concurrence of the ensemble $\left\{|\xi_1\rangle, |\xi_2\rangle\right\}$ that realizes~$\rho$  is 
\begin{equation}
	\langle C\rangle=\sum_{i=1}^2 |\tau'_{ii}|=\mu_1+\mu_2, 
\end{equation}
and we define the average preconcurrence of this ensemble as  
\begin{equation}
	\langle\tilde{C}\rangle:=\sum_{i=1}^2 \tau'_{ii} = \mu_1-\mu_2. \label{aveprecon}
\end{equation}
If the concurrence of the states~$|\xi_i\rangle$ are not equal, we can always interchange them by an orthogonal transformation.
Due to continuity, there must also be an intermediary orthogonal transformation $V$ that takes~$|\xi_i\rangle$ to states   
$|\zeta_i \rangle$
with the following property: $C_2\left(|\zeta_1\rangle \langle \zeta_1| \right)=C_2\left( |\zeta_2\rangle\langle \zeta_2|  \right)=\langle\tilde{C}\rangle$.   
Hence, the average concurrence of the new decomposition also equals the preconcurrence, 
\begin{equation}
	\langle C\rangle=\langle\tilde{C}\rangle=\mu_1-\mu_2.
\end{equation}

For any other decomposition of~$\rho$ attained by  the unitary operator  $V'$,
let $v_{ij}:=V'^2_{ij}$ so that $\sum_i |v_{ij}|=1$. The average concurrence is equal to   
\begin{align}
\langle C\rangle&=\sum_i \left| \sum_{j}  v_{ij} \tau'_{jj} \right|\ge \left|\mu_1-\sum_i v_{i2} \mu_2 \right|\nonumber\\
&\ge \mu_1-\left|\sum_i v_{i2} \mu_2\right|\ge\;\mu_1-\mu_2,
\end{align}
where we assume~$v_{i1}$ are real, by a suitable change of the overall phase if necessary, so that $\sum_i v_{i1}=1$. Thus, the average  concurrence of the ensemble~$\left\{|\zeta_1\rangle, |\zeta_2\rangle\right\}$ is the minimum average concurrence. 
\end{proof}
  
As for entanglement, we can use the closed form of the average concurrence to calculate the convex-roof extension of a set of other frameness monotones. The following corollary specifies what type of monotones belong to this set.  
\begin{corollary}
\label{cor:Cor8}
If a pure-state frameness measure~$F_\text{pure}$  is a non-decreasing convex  linear functional
$\mathcal{F}(C_2)$,  then $F(\rho)=\mathcal{F}(C_2(\rho))$ for all $\rho\in\mathcal{S}(\mathscr{H})$
with~$F$ the convex-roof extension of $F_\text{pure}$. 
\end{corollary}
\begin{proof}
Recall that there exists a decomposition with the minimum average concurrence where    
$C_2(\rho)=C_2(|\zeta_1\rangle\langle \zeta_1|)=C_2(|\zeta_2\rangle\langle \zeta_2|)$.
Thus, for any other decomposition of $\rho=\sum_{j} q'_{j} |\psi_{j}\rangle \langle\psi_{j}|$ we must  have~$C_2(|\zeta_1\rangle \langle \zeta_1|) \le \sum_{j} q'_{j} C_2(|\psi_{j}\rangle \langle \psi_{j}|)$.    
As~$\mathcal{F}(C_2)$ is non-decreasing and convex, we have  
\begin{align}
\mathcal{F}\left( C_2(|\zeta_1\rangle\langle \zeta_1|)\right)&\le\mathcal{F}\left( \sum_{j} q'_{j} \:C_2(|\psi_{j}\rangle \langle \psi_{j}|)\right)\nonumber\\
&\le \sum_{j} q'_{j} \:\mathcal{F}\left(C_2\left(|\psi_{j}\rangle \langle \psi_{j}|\right)\right).   
\end{align}
Thus $\mathcal{F}\left(C_2(\rho)\right)=\mathcal{F}\left(C_2(|\zeta_1\rangle\langle \zeta_1|)\right)$ is equal to the average frameness of~$\rho$ minimized over all its decompositions. 
\end{proof}
In the next section we use this corollary to calculate the convex-roof extension of the variance~(\ref{eq:Var}),
which is the asymptotic measure  of U(1)-frameness.
Finally, we note that this corollary can also be used for the convex-roof extension of the asymptotic $\mathbb{Z}_2$-frameness measure~\cite{GS08}.

\section{Frameness of Formation}
\label{sec:FoF}

Proposition~\ref{Vidalf} enables us to systematically construct frameness monotones under U(1)-SSR, but as we noted earlier, not all U(1)-frameness monotones can be obtained this way.
Yet, there is still a chance that the convex-roof extension of  monotones that cannot be obtained by the method of Proposition~\ref{Vidalf} may be expressed directly as a function of monotones that do.  In particular,  Corollary~\ref{cor:Cor8} specifies which monotones can be related in this way to the concurrence of frameness for the case of single qubit states. 

The number variance~(\ref{eq:Var}) is an important frameness monotone that does not meet the conditions of Proposition~\ref{Vidalf},~i.\ e.\ it  is not a concave function of the twirled state.          
Besides being an ensemble monotone, variance is the unique measure of frameness of pure states
in the sense that it quantifies the rate at which they can be asymptotically formed from or distilled into the state $|+\rangle:=\left(|0\rangle+|1\rangle\right)/\sqrt{2}$~\cite{GS08}. The  $|+\rangle$ state is chosen as a standard unit resource state and is an instance of a unipartite, or local, refbit~\cite{Enk05, GS08}.  Thus, the convex-roof extension of the variance is  the equivalent of the entanglement of formation~\cite{BDSW96} and is therefore called the frameness of formation (FoF) of the group U(1)~\cite{GS08}. 
\begin{definition}
The {\it frameness of formation} for the group  $G$=U(1) of a state~$\rho$ in terms of  refbits $|+\rangle$  is 
\begin{equation}
V(\rho)\equiv\min
_{\left\{|\psi_i\rangle, q_i\right\}} \sum_i q_i \; V_\text{pure}(|\psi_i\rangle\langle \psi_i|),\label{FF}
\end{equation}
where $V_\text{pure}(|\psi_i\rangle\langle \psi_i|)=4\left(\langle\psi_i|\hat{n}^2|\psi_i \rangle- \left\langle\psi_i \right|\hat{n}\left| \psi_i \right>^2\right)$. 
\end{definition}
\noindent
As we presently show, the variance of a qubit is a convex function of the concurrence,  and we can determine the FoF  of a qubit analytically by relating the variance to the qubit's concurrence of frameness using Corollary~\ref{cor:Cor8}.     
The outcome is analogous to Wootters's formula for the entanglement of formation of bipartite two-qubit states~\cite{Woo98}. 
\begin{proposition}
The FoF of a single qubit is 
\begin{align}
\label{eq:VarFor}
V(\rho)=|\mu_1-\mu_2|^2
\end{align}
for~$\mu_R=\left\{\mu_1, \mu_2\right\}$ the set of eigenvalues for state~$R= \sqrt{\sqrt{\rho}\tilde{\rho}\sqrt{\rho}}$.
\end{proposition}
\begin{proof}
Recall that 
$|\psi\rangle=\sqrt{\lambda_{0}}|0\rangle+\sqrt{\lambda_1}|1\rangle$.
We have 
$C_2(|\psi\rangle\sl \langle \psi|)=2\sqrt{\lambda_0\lambda_1}$ 
 and  
$V_\text{pure} (|\psi\rangle\langle \psi|)=4\lambda_{0}\lambda_1$. 
The frameness~(\ref{eq:Var}) is a convex, non-decreasing function of the concurrence,  
$\mathcal{V}(C_2)=C_2^2$.
The result follows from Corollary \ref{cor:Cor8}.
\end{proof}

Schuch et.\ al.\  already identified the equivalent of the number variance in a bipartite setting  as a separate measure of non-locality under the joint restrictions of LOCC and total-number SSR, where they call the bipartite measure the ``superselection-induced variance" and also show, among other things, that its convex-roof extension can be obtained from the bipartite entanglement concurrence~\cite{SVC04}.  

The SSR-induced  variance and U(1)-fameness of formation are related, and the arguments in section~\ref{qudit-mono} that relate frameness resources to entanglement through purification of the twirled state makes the link between the two measures even more explicit. However, although we employ bipartite states for purification, our aim is not to study nonlocal frameness.
Rather,  the resources we consider are unipartite and are not restricted by this SSR. Only the operations have to obey the SSR. The monotones and measures of frameness  we consider, including the concurrence of frameness and the variance,  are viewed as local RF resources and are treated on their own, independent of entanglement theory.   

Strictly speaking, the variance quantifies the rate of formation for states $|\psi\rangle=\sum_n \lambda_n |n\rangle$ whose number spectrum, $\text{spec}(|\psi\rangle)$~(\ref{Spec}), is gapless,~i.\ e.\ states for which $\lambda_{n_1}>0$ and~$\lambda_{n_2}>0$ implies that $\lambda_n>0$ for all $n$ between $n_1$ and~$n_2$. The reason is that states with gaps cannot be transformed to gapless states  with non-zero probability under U(1)-SSR. However, the problem can be solved by employing  negligible amount of catalyst resources that makes it possible to asymptotically transform gapless and gapped states to each other in a reversible manner~\cite{SVC04,GS08}. 
Here, we have assumed all pure states are mutually interconvertible. Under this assumption, we can consistently interpret the convex-roof extension of the variance as the minimum average cost, in terms of refbits,  of preparing the ensemble of states that realize the mixed state.
    
 \section{Conclusions}

We have developed a framework for studying frameness measures and studied in detail the
special case that the lack of reference frame information corresponds to ignorance about a phase reference,
i.e.\ a U(1) invariance.
Our strategy has been to adapt existing entanglement monotones for mixed states, 
which use the convex-roof extension,
to the case of frameness monotones.
The concurrence frameness illustrates this point well: concurrence of entanglement
is straightforwardly adapted to a frameness concurrence monotone using a convex-roof extension of the pure-state case.
We also introduced a Ôframeness of formationÕ monotone.
This monotone quantifies the number of refbits required to construct a resource state,
analogously to the case of entanglement whereby ebits are consumed to construct a resource states.
We show that the resultant frameness of formation is indeed a proper monotone.

In the restricted U(1)-invariance case considered here, we have seen that,
despite the great difference between superselection rules and local operations under classical communication,
with the former related to frameness and the latter to entanglement,
  every pure state bipartite entanglement resource has an analogous frameness resource.

This connection is explicitly clear in our strategy of purifying the twirled state, which results in a bipartite system.     
Whether  this connection between frameness and entanglement holds  beyond U(1) invariance,
and specifically for non-abelian group invariance, and whether other entanglement monotones that are not convex-roof extensions  are similarly related to frameness monotones  are interesting questions that remain open.    
Our work on U(1) invariance establishes a foundation for studies of quantum frameness for mixed-state resources
and suggests deeper connections between entanglement and frameness measures.

\begin{acknowledgments}
We appreciate valuable discussions with  I.\ Marvian, Y.\ R.\ Sanders,  M.\ Skotiniotis, and R.\ W.\ Spekkens.  
This research has been supported by Alberta Innovates,   
the Natural Sciences and Engineering Research Council,
General Dynamics Canada,
and the Canadian Centre of Excellence for Mathematics of Information Technology and Complex Systems (MITACS).
BCS is partially supported by a Canadian Institute for Advanced Research Fellowship.
\end{acknowledgments}
\appendix
\section{Explicit concurrence of frameness and U(1)-frameness of formation of a qubit}
\label{sec:Ap1}
Let $\rho=p|\phi_1\rangle\langle\phi_1|+(1-p)|\phi_2\rangle\langle\phi_2|$
be the spectral decomposition of the state~$\rho$. The two states $|\phi_1\rangle$ and $|\phi_2\rangle$ have the same relative phase and can be simultaneously transformed by U(1)-invariant transformations to states with real amplitudes on the Bloch sphere:   
\begin{align}
	|\phi_1\rangle=&\cos\frac{\alpha}{2}|0\rangle+\sin\frac{\alpha}{2} |1\rangle,\nonumber\\
	|\phi_2\rangle=&-\sin\frac{\alpha}{2}|0\rangle+\cos\frac{\alpha}{2} |1\rangle.
\end{align}
The two singular values of the state~$R(\rho)$ in Eq.~(\ref{R}) are 
\begin{align}
	\mu_{1,2}=\sqrt{p(1-p)+\frac{1}{2}\left(1-2p\right)^{2}\sin^2\alpha \pm K}
\end{align}
for
\begin{align}
K:=\frac{1}{2} \big|\left(1-2p\right)\sin\alpha \big| \sqrt{\left(1-2p\right)^2\sin^2\alpha +4p(1-p)}. 
\end{align}
The state's concurrence is equal to 
\begin{align}
C_2(\rho)=\left|(1-2p) \sin \alpha\right|,  
\end{align}
and  
\begin{align} 
 V_2(\rho)&=(1-2p)^2 \sin^2 \alpha  
 \end{align}
is the qubit U(1)-frameness of formation.


\begin{thebibliography}{00}
\bibitem{BRS07}
	S.\ D.\ Bartlett, T.\ Rudolph and R.\ W.\ Spekkens,
	Rev. Mod. Phys. \textbf{79}, 555 (2007). 
\bibitem{BRST09}
	S.\ D.\ Bartlett, T.\ Rudolph, R.\ W.\ Spekkens and P.\ S.\ Turner,
	New J.  Phys. \textbf{11}, 063013 (2009). 	
\bibitem{Enk05}
	S.\ J.\ van Enk,
	\pra \textbf{71}, 032339 (2005).
\bibitem{BRS03}
	S.\ D.\ Bartlett, T.\ Rudolph and R.\ W.\ Spekkens, 
	\prl \textbf{91}, 027901 (2003).   		
\bibitem{SVC04}
	N.\ Schuch, F.\ Verstraete, and J.\ I.\ Cirac,
	\prl \textbf{92}, 087904 (2004); \pra \textbf{70}, 042310 (2004).
\bibitem{GS08}
	G.\ Gour and R.\ W.\ Spekkens,
	New J.  Phys. \textbf{10} 033023 (2008). 	
\bibitem{GMS09}
	G.\ Gour, I.\ Marvian and R.\ W.\ Spekkens,
	\pra \textbf{80}, 012307 (2009).
\bibitem{Mar11}
	I.\ Marvian and R.\ W.\ Spekkens, arXiv:1104.0018.
\bibitem{Vid00}
	G.\ Vidal, 
	J. Mod. Opt. \textbf{47}, 355 (2000).
\bibitem{Ple05}
	M.\ B.\ Plenio, 
	\prl \textbf{95}, 090503 (2005).
\bibitem{Gou05a}
	G.\ Gour, 
	\pra \textbf{71}, 012318 (2005).
\bibitem{Gou05b}
	G.\ Gour, 
	\pra \textbf{72}, 042318 (2005).
\bibitem{Gou06}
	G.\ Gour, 
	\pra \textbf{74}, 052307 (2006).
\bibitem{RBC+01}
	P.\ Rungta, V.\ Bu\v{z}ek, C.\ M.\ Caves, M.\ Hillery and G.\ J.\ Milburn,
	\pra \textbf{64}, 042315 (2001).
\bibitem{BDSW96}
	C.\ H.\ Bennett, D.\ P.\ DiVincenzo, J.\ A.\ Smolin and W.\ K.\ Wootters, 
	\pra \textbf{54},  3824 (1996).
\bibitem{Woo98}
	W.\ K.\ Wootters,
	\prl  \textbf{80}, 2245 (1998).
\end{thebibliography}
\end{document}